\documentclass[12pt]{article}
\usepackage{hyperref}
\usepackage{amsthm}
\usepackage{latexsym}
\usepackage[latin1]{inputenc}
\usepackage{amsfonts}
\usepackage{hhline}
\usepackage{graphics}
\usepackage{graphicx}
\usepackage{amsmath}

\usepackage{amssymb}
\usepackage{makeidx}
\usepackage{hyperref}

\makeindex
\hypersetup{colorlinks=true, citecolor=black, urlcolor=black, linkcolor=black}

\def\indicator{\mathbf{1}}

\def\E{\operatorname{\bf E}}
\def\P{\operatorname{\bf P}}

\def\Q{\operatorname{\bf Q}}
\def\1{{\bf 1}}
\def\I{{\mathcal I}}
\def\R{\mathbb{R}}
\def\Z{\mathbb Z}
\def\C{\mathbb C}

\def\V{\mathcal V}

\def\beq{\begin{equation}}
\def\eeq{\end{equation}}

\def\process#1{#1=\{#1_t\}_{t\geq 0}}

\newtheorem{lemma}{Lemma}[section]

\makeatletter
\newcommand{\rfill}[1]{\ifmeasuring@#1\else\omit\hfill$\displaystyle#1$\fi\ignorespaces}
\newcommand{\lfill}[1]{\ifmeasuring@#1\else\omit$\displaystyle#1$\hfill\fi\ignorespaces}
\makeatother

\usepackage{natbib}
 \def\cite#1{(\citet{#1})}
\begin{document}
\title{Computing Greeks for L\'evy Models: The Fourier Transform Approach.}
\author{Federico De Olivera\thanks
{Mathematics Center,  School of Sciences, Universidad de la
Rep\'ublica, Montevideo. Uruguay. e-mail: fededeo@gmail.com.}\hspace{0.2cm} Ernesto Mordecki\thanks
{Mathematics Center,  School of Sciences, Universidad de la
Rep\'ublica, Montevideo. Uruguay. e-mail: mordecki@cmat.edu.uy.}}
\date{\today }
\maketitle

\begin{abstract}
The computation of Greeks for exponential L\'evy models are usually approached by Malliavin Calculus and other methods, as the Likelihood Ratio and the finite difference method. 
In this paper we obtain exact formulas for Greeks of European options based on the  Lewis formula for the option value. 
Therefore, it is possible to obtain accurate approximations using Fast Fourier Transform.  We will present an exhaustive development of Greeks for Call options. The error is shown for all Greeks in the Black-Scholes model, where Greeks can be exactly computed. Other models used in the literature are compared, such as the Merton and Variance Gamma models. The presented formulas can reach desired accuracy because  our approach generates error only by approximation of the integral.
\newline {\bf Keywords:} Greeks; exponential L\'evy models.
\end{abstract}

 \newpage
 \section{Introduction}


We consider a L\'evy process $\process X$ defined on a probability
space $(\Omega, {\cal  F}, \Q)$, 
 a finantial market model with two assets, a deterministic savings
account $\process B$, given by \[B_t=B_0e^{rt},\] with $r\ge 0$ and $B_0>0$,  and a stock $\process S$, given by
\begin{equation}\label{eq:stock}
S_t=S_0e^{rt+X_t},
\end{equation}
with $S_0>0$, where $\process X$ is a L\'evy process. 
When the process $X$ has continuous paths,
we obtain the classical Black-Scholes model \cite{Merton73}. 
For general reference on the
subject we refer to \cite{Kyprianou06} or \cite{CT04}.


The aim of this paper is the computation of the price partial derivatives of an European option
with general payoff 
with respect to any parameter of interest. These derivatives
are usually named as ``Greeks'', and consequently we use the term \emph{Greek} to refer to any price parcial derivative 
 of the option (of any order and with respect to any parameter).

Our approach departs from the subtle observation by Cont and Tankov see (\citet{CT04}, p. 365):

\emph{``Contrary to the classical Black-Scholes case, in exponential-L\'evy models
there are no explicit formulae for call option prices, because the probability
density of a L\'evy process is typically not known in closed form. However, the
characteristic function of this density can be expressed in terms of elemen-
elementary functions for the majority of L\'evy processes discussed in the literature.
This has led to the development of Fourier-based option pricing methods for
exponential-L\'evy models. In these methods, one needs to evaluate one Fourier
transform numerically but since they simultaneously give option prices for a
range of strikes and the Fourier transform can be efficiently computed using
the FFT algorithm, the overall complexity of the algorithm per option price
is comparable to that of evaluating the Black-Scholes formula.''}

In other words, in the need of computation of a range of option prices,
from a practical point of view, the Lewis formula works as a closed formula,
as it can be implemented and computed with approximately the same precision and in 
the same time as the Black Scholes formula.

Some papers have addressed this problem. Eberlein, Glau and Papapantoleon \cite{0809.3405v4}  obtained 
a formula similar to the Lewis one, and derived delta ($\Delta$) and gamma ($\Gamma$), the price partial  derivatives with respect to the initial value $S_t$ of first and second order, for an European payoff function. The assumptions are  similar to the ones we require.

Takahashi and Yamazaki \cite{Taka} also obtain these Greeks in the case of  Call options, based on the Carr and Madan approach. The advantage of the Lewis formula is that it gives option prices for general European payoffs,
while Carr-Madam only price European vanilla options.

 Other works deal with the problem of Greeks computation for more general payoff functions, 
 including path dependent options, for example see \cite{Chen2007},\cite{Glasserman2007},\cite{Glasserman08estimatinggreeks}, \cite{kienitz}, \cite{Boya},  \cite{jeannin2010transform}. These works are based on different techniques, such
 as simulation or finite differences introducing a method error, that has to be analyzed whereas our approach does not.

In the present paper we obtain closed formulas for Greeks based on the Lewis formula, that computes efficiently
and with arbitrary precision (as exposed in \cite{CT04}), for arbitrary payoff European options in the L\'evy models with respect to any parameter and arbitrary order. As an example we analyze the case of Call options.

\section{Greeks for General European Options in Exponential L\'evy Models}

In this paper we do not address the interesting problem of the determination of the pricing measure, see for instance \cite{CT04}, assuming then that the given measure $\Q$ is the risk-neutral pricing measure.
In other words, we assume that the martingale condition is satisfied under $\Q$, that in view of \eqref{eq:stock},
stands for the condition 
\begin{equation}\label{eq:martingale}
\E e^{X_1}=1.
\end{equation}

Furthermore, by the L\'evy-Khinchine Theorem, we obtain that $\E e^{izX_t}=e^{t\Psi(z)}$, where the characteristic exponent is
\begin{equation}\label{eq:psi}
\Psi(z)=-iz(1-iz)\frac{\sigma^2}{2}+\int_\R\big(e^{izy}-1-iz(e^y-1)\big) \nu(dy),
\end{equation}
with $\sigma\geq 0$ is the standard deviation of the gaussian part of the L\'evy process, and $\nu$ its jump measure.



Regarding the payoff, following \cite{lewis01}, denote $s=\ln S_T$ and consider a payoff $w(s)$. For instance,
if $K$ is a strike price, 
\begin{equation}\label{eq:cp} 
w(s)=(e^s-K)^+
\end{equation} 
is the call option payoff.
Then, being $\widehat{w}(z)$ the Fourier Transform of $w(s)$, the Lewis formula \cite{lewis01}
for the European options, valued at time $t$, and denoting $\tau=T-t$ the time to maturity, is:
\begin{equation}
 V_t=\frac{e^{-r\tau}}{2\pi}\int_{iv+\R}e^{-iz(\ln(S_t)+r\tau)}e^{\tau \Psi(-z)}\widehat{w}(z)dz, \label{eq:option}
\end{equation}
where $z\in S_V=\{u+iv\colon u\in\R\}$ and $v$ must be chosen depending on the payoff function \cite{lewis01}.
In this context, it is simple to obtain some general formulas for the Greeks.

In order to differentiate under the integral sign, we present the following clasical result.
\begin{lemma} \label{lemma1}
Let $\Theta\subset \R$ an interval and $\I=iv+\R$. Let $h:\I\times\Theta\to \C$ and $g:\I\to \C$ such that 
\begin{itemize}
 \item $h(\cdot, \theta)g(\cdot)$ is integrable for all $\theta \in \Theta$ and $g$ is integrable.
 \item $h(z, \cdot)$ is differentiable in $\Theta$ for all $z\in \I$ and $\frac{\partial h}{\partial \theta}$ is bounded.
 \end{itemize}

Then, $\int_\I h(x,\theta)g(x)dx$ is differentiable  and $$\frac{\partial}{\partial \theta} \int_\I h(x,\theta)g(x)dx=\int_\I \frac{\partial h(x,\theta)}{\partial \theta}g(x)dx\qquad \forall \theta \in \Theta.$$  

\end{lemma}

\begin{proof}
We observe that $\left|\frac{\partial h(z,\theta)g(z)}{\partial \theta}\right|\leq C |g(z)|$ for all $z\in \I$, $\theta \in \Theta$. The result is obtained from Theorem 2.27 in  \cite{folland1999real}.  
\end{proof}
 
In consequence, in what follows, we will always assume that the conditions in Lemma \ref{lemma1} are satisfied for the real part of the integrand because the price imaginary part integrate is zero.

\subsection{First Order Greeks}

We introduce the auxiliary function $$\vartheta(z)=e^{-iz(\ln(S_t)+r\tau)}e^{\tau\Psi(-z)}\widehat{w}(z).$$
Departing from \eqref{eq:option} and \eqref{eq:psi}, by differentiation under the integral sign we obtain  
{\allowdisplaybreaks\begin{align*}
 \Delta_t=&\frac{\partial V_t}{\partial S_t}= -\frac{1}{S_t}\frac{e^{-r\tau}}{2\pi}\int_{iv+\R}iz\vartheta(z)dz,  \\
 \rho_t=&\frac{\partial V_t}{\partial r}= -\tau \frac{e^{-r\tau}}{2\pi}\int_{iv+\R}(1+iz)\vartheta(z)dz,  \\
 \V_t=&\frac{\partial V_t}{\partial \sigma}= \tau\sigma \frac{e^{-r\tau}}{2\pi}\int_{iv+\R}iz(1+iz)\vartheta(z)dz, \\
 \Theta_t=&\frac{\partial V_t}{\partial \tau}= \frac{e^{-r\tau}}{2\pi}\int_{iv+\R}\big[\Psi(-z)-(1+iz)r\big]\vartheta(z)dz. 
\end{align*}}

Usually, the L\'evy models used in the literature depend on a set of parameters, that specify the jump measure. 
Therefore we denote $\nu(dy)=\nu_\theta(dy)$ and $\Psi(z)=\Psi_\theta(z)$, then:

\begin{align*}
\frac{\partial V_t}{\partial \theta}= \tau\frac{e^{-r\tau}}{2\pi}\int_{iv+\R}\frac{\partial \Psi_\theta(-z)}{\partial \theta}\vartheta(z)dz. 
\end{align*}

\subsection{Second Order Greeks}
Similarly, we obtain
{\allowdisplaybreaks\begin{align*}
 \Gamma_t=&\frac{\partial^2 V_t}{\partial S_t^2}= \frac{1}{S_t^2}\frac{e^{-r\tau}}{2\pi}\int_{iv+\R}iz(1+iz)\vartheta(z)dz,  \\
 \text{Vanna}_t=&\frac{\partial^2 V_t}{\partial \sigma\partial S_t}=\tau \sigma \frac{1}{S_t}\frac{e^{-r\tau}}{2\pi}\int_{iv+\R}z^2(1+iz)\vartheta(z)dz, \\
\text{Vomma}_t=&\frac{\partial^2 V_t}{\partial \sigma^2}=\tau \frac{e^{-r\tau}}{2\pi}\int_{iv+\R}\big[1+\tau\sigma^2iz(1+iz)\big]iz(iz+1)\vartheta(z)dz, \\
\text{Charm}_t=&\frac{\partial^2 V_t}{\partial S_t\partial \tau}= -\frac{1}{S_t}\frac{e^{-r\tau}}{2\pi}\int_{iv+\R}iz\big[\Psi(-z)-(1+iz)r\big]\vartheta(z)dz, \\
\text{Veta}_t=&\frac{\partial^2 V_t}{\partial \sigma\partial \tau}= \sigma\frac{e^{-r\tau}}{2\pi}\int_{iv+\R}iz(1+iz)\big[\tau\Psi(-z)-(iz+1)r\tau+1\big]\vartheta(z)dz, \\
\text{Vera}_t=&\frac{\partial^2 V_t}{\partial\sigma\partial r}= -\tau^2\sigma \frac{e^{-r\tau}}{2\pi}\int_{iv+\R}iz(iz+1)^2\vartheta(z)dz. 
\end{align*}}

Other derivatives can be  obtained analogously. 
In next section we will focus in the case of Call options. This allows to obtain  more explicit formulas.

\section{Greeks for Call Options in Exponential L\'evy Models} \label{GCALL}

In order to exploit the particular payoff function, we  exhaustively develop the Greeks for  Call options. The Put option corresponding formulas can be obtained immediately  via Put-Call parity. For other payoff the procedure to obtain the Greeks is analogous. 

When the strike $K$ is fixed, $x=\ln(K/S_t)-r\tau$ is variable in terms of $S_t$, $r$ and $\tau$. Then, we must consider this for the computation of Greeks $\Delta$, $\Gamma$, $\rho$ and others. 


\begin{lemma} \label{lemma3}
 Let $X_\tau$ be a L\'evy process with triplet $(\gamma, \sigma,\nu)$ and charateristic exponent $\Psi(z)$ such that $\Psi(-i)=0$ and $\int_{|y|>1}e^{vy}\nu(dy)<\infty$ with $v\geq 0$. Then, if $z\in iv+\R$
 
  \begin{align*}|\Psi_J(-z)| \leq (|z|^2+|z|)\frac{e^{v}}{2}\int_{|y|\leq 1}  y^2  \nu(dy)+2\int_{|y|>1} (e^{vy}+1) \nu(dy)
       \end{align*}
       
      and 
      
   \begin{align}|\Psi(-z)| \leq (|z|^2+|z|)\Big(\frac{e^{v}}{2}\int_{|y|\leq 1}  y^2  \nu(dy)+\frac{\sigma^2}{2}\Big)+2\int_{|y|>1} (e^{vy}+1) \nu(dy), \label{cotaPsi}
   \end{align}   
 where $\Psi_J(z)=\int_\R \big[e^{izy}-1-iz(e^y-1) \big]\nu(dy)$.  
        
\end{lemma}

\begin{proof}
 
Let $I(z)=\int_\R \big[e^{izy}-1-izy\indicator_{\{|y|\leq 1\}} \big]\nu(dy)$. 
Applying Taylor's expansion with Lagrange error form at point $y=0$, 
there exists $\theta_y$ with  $|\theta_y|\leq |y|$ such that

\begin{align*}
 e^{izy}-1-izy\indicator_{\{|y|\leq 1\}}= &izy\indicator_{\{|y|> 1\}}- z^2y^2\frac{e^{iz\theta_y}}{2} \\
 =&- z^2y^2\frac{e^{iz\theta_y}}{2} \indicator_{\{|y|\leq 1\}} + \big(izy-z^2y^2\frac{e^{iz\theta_y}}{2}\big)\indicator_{\{|y|> 1\}}\\
 =&-z^2y^2\frac{e^{iz\theta_y}}{2} \indicator_{\{|y|\leq 1\}}+ (e^{izy}-1)\indicator_{\{|y|> 1\}}.
\end{align*}

Then 
\begin{align}
 |I(-z)|\leq&\int_\R \Big| -z^2y^2\frac{e^{-iz\theta_y}}{2} \indicator_{\{|y|\leq 1\}}+ (e^{-izy}-1)\indicator_{\{|y|> 1\}} \Big|\nu(dy)\notag\\
 \leq & |z|^2\int_{|y|\leq 1}  \frac{e^{v\theta_y}}{2}y^2  \nu(dy)+\int_{|y|>1} (e^{vy}+1) \nu(dy)\notag\\
 \leq & |z|^2\frac{e^{v}}{2}\int_{|y|\leq 1}  y^2  \nu(dy)+\int_{|y|>1} (e^{vy}+1) \nu(dy)\label{cotaI}.
 \end{align}

Using \eqref{cotaI} we have  \begin{align*}|\Psi_J(-z)|=&|I(-z)+izI(-i)|\\ \leq &(|z|^2+|z|)\frac{e^{v}}{2}\int_{|y|\leq 1}  y^2  \nu(dy)+2\int_{|y|>1} (e^{vy}+1) \nu(dy).
       \end{align*}

For the continuous part, let  $\Psi_C(-z)=(iz-z^2)\frac{\sigma^2}{2}$, thus
 
 \begin{align*}|\Psi(-z)|\leq &|\Psi_C(-z)|+|\Psi_J(-z)|\\
 \leq& (|z|^2+|z|)\Big(\frac{e^{v}}{2}\int_{|y|\leq 1}  y^2  \nu(dy)+\frac{\sigma^2}{2}\Big)+2\int_{|y|>1} (e^{vy}+1) \nu(dy).
   \end{align*}

\end{proof}

\begin{lemma}\label{derivadasLewis}
 Let $\{X_\tau\}_{\tau\geq 0}$ be a L\'evy process with triplet $(\gamma, \sigma,\nu)$ and charateristic exponent $\Psi(z)$, such that $\Psi(-i)=0$ and $\E[e^{vX_\tau}]<\infty$ with $v> 0$. 
 
 \begin{enumerate}
 \item If $\int_{iv+\R} |z|^{-1}|e^{\tau\Psi(-z)}|dz<\infty$ then 
 \begin{align}
  \P(X_\tau>x)=&-\frac{1}{2\pi}\int_{iv+\R}\frac{e^{izx}}{iz}e^{\tau\Psi(-z)}dz, \label{PrLewis}\\
  \E(e^{X_\tau}\indicator_{\{X_\tau>x\}})=&-\frac{1}{2\pi}\int_{iv+\R}\frac{e^{(1+iz)x}}{1+iz}e^{\tau\Psi(-z)}dz. \label{Pr2Lewis}
 \end{align}
  \item  If $\int_{iv+\R} |z|^n|e^{\tau\Psi(-z)}|dz<\infty$ for some $n\in \Z$, then $X_\tau$ has a density of class $C^n$ and
 \begin{align}
  \frac{\partial ^n f(x)}{\partial x^n}=\frac{1}{2\pi}\int_{iv+\R}(iz)^ne^{izx}e^{\tau\Psi(-z)}dz. \label{ddens}
 \end{align}
\end{enumerate}

\begin{proof} For a Call option the Fourier Transform of the payoff function is $\widehat{w}(z)=\frac{e^{iz(\ln(S_t)+r\tau+x)}}{iz(1+iz)}$. Then from the option value \eqref{eq:option} we have, with  $x=\log(K/S_t)-r\tau$,  
\begin{equation}\label{eq:CallLewis}
C_t(x)=S_t\frac{e^{x}}{2\pi}\int_{iv+\R}\frac{e^{izx}e^{\tau \Psi(-z)}}{iz(iz+1)}dz.
\end{equation}
Then, being $x\in [\alpha, \beta]$ and $C_1=\max_{\alpha\leq x\leq \beta} e^{(1-v)x}$
\begin{equation*}\left|\frac{\partial e^{(1+iz)x}e^{\tau \Psi(-z)}[iz(iz+1)]^{-1} }{\partial x}\right|\leq C_1|z^{-1}||e^{\tau\Psi(-z)}|\in L^1(iv+\R), 
\end{equation*}
and by Theorem 2.27 in \cite{folland1999real} we can differentiate under the integral sign. Therefore, with $S_t=1$ 

\begin{align*}\P(X_\tau>x)=& -e^{-x}\frac{\partial }{\partial x}\int_{x}^{\infty}(e^s-e^x)F(ds)\\
=&-e^{-x}\frac{\partial C_t(x)}{\partial x}=-\frac{1}{2\pi}\int_{iv+\R}\frac{e^{izx}}{iz}e^{\tau\Psi(-z)}dz.  
\end{align*}
On the other hand, with $S_t=1$ 

\begin{align*}
 \E(e^{X_\tau}\indicator_{\{X_\tau>x\}})=&\int_{x}^{\infty}e^s F(ds)=-e^{x}\frac{\partial }{\partial x}\int_x^{\infty}(e^{s-x}-1)F(ds)\\
 =&-e^x \frac{\partial e^{-x}C_t(x)}{\partial x}=C_t(x)-\frac{\partial C_t(x)}{\partial x}\\
 =&\frac{1}{2\pi}\int_{iv+\R}\frac{e^{(1+iz)x}}{iz(1+iz)}e^{\tau \Psi(-z)}dz -\frac{1}{2\pi}\int_{iv+\R}\frac{e^{(1+iz)x}}{iz}e^{\tau \Psi(-z)}dz\\
 =&-\frac{1}{2\pi}\int_{iv+\R}\frac{e^{(1+iz)x}}{1+iz}e^{\tau \Psi(-z)}dz.
\end{align*}
For the second part, observe in \eqref{PrLewis} that  if  $x\in [\alpha, \beta]$ and $C_2=\max_{\alpha\leq x\leq \beta} e^{-vx}$ \begin{align*}\left|\frac{\partial^{n+1} \frac{e^{izx}}{iz}e^{\tau\Psi(-z)} }{\partial x^{n+1}}\right|\leq C_2|z|^n|e^{\tau\Psi(-z)} |\in L^1(iv+\R).
\end{align*}
The result is obtained from Theorem 2.27 in \cite{folland1999real}.
\end{proof}

\end{lemma}

%


\subsection{Generalized Black-Scholes Formula for L\'evy Processes.}

We consider by now that $\process X$ in \eqref{eq:stock} is an arbitrary stochastic process satisfying the martingale condition \eqref{eq:martingale}, 
and we introduce the measure $\tilde{\Q}$ by the equation
\begin{equation}\label{eq:esscher}
\frac{d\tilde\Q}{d\Q}=e^{X_T}.
\end{equation}
This new measure $\tilde\Q$ is the Esscher Transform of $\Q$ with parameter $\theta=1$, 
and it was baptized by Shiryaev et al. \cite{ShiryaevKabanovKramkov94} as the \emph{dual martingale measure}.

We consider a call option with payoff \eqref{eq:cp}, and denote the log forward moneyness\footnote{This seems to be the standard definition, although in \cite{CT04} is defined as the opposite quantity.} by $x=\ln(K/S_0)-r\tau$. Its price in the model we consider can be transformed as 
\begin{align*}
C_t(x)	&=e^{-r\tau}\E(S_te^{r\tau+X_\tau}-S_te^{r\tau+x})^+
					=S_t\E(e^{X_\tau}-e^x)^+ \notag \\
			&=S_t\E(e^{X_\tau}-e^x)\indicator_{\{X_\tau>x\}}
					=S_t\left(\E e^{X_\tau}\indicator_{\{X_\tau>x\}}-e^x\E \indicator_{\{X_\tau>x\}}\right) \notag\\
			&=S_t\left(\tilde\Q(X_\tau>x)-e^x\Q(X_\tau>x)\right).
\end{align*}

Then, we have a closed formula in terms of the probability $\Q$ and $\tilde{\Q}$. This formula is obtained in (\citet{tankov2010financial}, p. 68) and is a generalization of the Black-Scholes formula when the underlying asset $X_T$ is a normal random variable.
Furthermore we observe  that the first term to be computed
$$
\E e^{X_\tau}\indicator_{\{X_\tau>x\}}=\tilde\Q(X_\tau>x)
$$
is the price of an \emph{asset or nothing} option, while the second term
$$
\E \indicator_{\{X_\tau>x\}}=\Q(X_\tau>x)
$$
is the price of a \emph{digital} option.


In the  case that   $\process X$ is a L\'evy process under $\Q$, we obtain the characteristic triplet $(\tilde{\sigma}^2, \tilde\nu, \tilde \gamma)$ under $\tilde\Q$ by the formulas
\begin{align*}
\tilde{\sigma}&=\sigma, \\
\tilde{\nu}(dx)&=e^x \nu(dx), \\
\tilde{\gamma}&=\frac{\sigma^2}{2}+\int_\R(e^{-y}-1+h(y))\tilde{\nu}(dy).
 \end{align*}
Furthermore, if $X_t$  has a density $f_t(x)$, 
by \eqref{eq:esscher}, we obtain the density $\tilde{f}_t$ of $X_t$ under $\tilde{\Q}$, given by
\begin{align*} 
 \tilde{f}_t(s)=e^s f_t(s).
\end{align*}
%
%
%
%




In order to obtain Greeks in terms of the risk neutral measure, we replace $\P$ by $\Q$ in \eqref{PrLewis} and consequently \eqref{Pr2Lewis}, \eqref{ddens} and \eqref{eq:CallLewis} are related to the probability measure $\Q$.

%
%
%
%
%
%
%
%

\subsection{First Order Greeks for Call Options}

In this section we do not assume general requirements. 
We specify the requirements in each case. 


\subsubsection*{Delta}

Assume that $\int_{iv+\R}|z|^{-1}|e^{\tau \Psi(-z)}|dz<\infty$ and  $S_t\in [A,B]$.
From \eqref{eq:CallLewis} we obtain

\begin{align*}
 \Delta_{t}^{L}=&\frac{\partial C_{t}(x(S_t))}{\partial S_t}\notag\\
 =& \frac{\partial }{\partial S_t}S_t\frac{1}{2\pi}\int_{iv+\R} \frac{e^{(1+iz)x(S_t)}}{iz(1+iz)}e^{\tau\Psi(-z)}dz \notag\\
 =& \frac{1}{2\pi}\int_{iv+\R} \frac{e^{(1+iz)x}}{iz(1+iz)}e^{\tau\Psi(-z)}dz-\frac{1}{2\pi}\int_{iv+\R} \frac{e^{(1+iz)x}}{iz}e^{\tau\Psi(-z)}dz \notag\\
 =& -\frac{1}{2\pi}\int_{iv+\R} \frac{e^{(1+iz)x}}{1+iz}e^{\tau\Psi(-z)}dz=\tilde{\Q}(X_\tau>x). 
 \end{align*}

%
%

\subsubsection*{Rho}

Denote now $x=\ln(K/S_t)-r\tau$, that depends on the interest rate $r$.  
Assume that $\int_{iv+\R}|z|^{-1}|e^{\tau \Psi(-z)}|dz<\infty$ and  $r\in [R_1,R_2]$.
Then


\begin{align*}
 \rho_{t}^L=&\frac{\partial C_{t}(x(r))}{\partial r}= S_t\frac{1}{2\pi}\int_{iv+\R}-\tau\frac{e^{(1+iz)x}}{iz}e^{\tau\Psi(-z)}dz =\tau S_te^x\Q(X_\tau>x).
  \end{align*}


\subsubsection*{Vega} In Black-Scholes, Vega shows the change in variance of the log-price. 
In exp-L\'evy models, the derivative of $C_t(x)$ w.r.t. $\sigma$ does not give exactly the same information.
 We assume that $X_\tau$ has density $f$, $\sigma\in [\Sigma_1,\Sigma_2]$ with $\Sigma_1>0$ and $z\in iv+\R$. Let 
 \begin{align*}h(z,\sigma)=&{e^{\tau iz(1+iz)\frac{\sigma^2}{2}}}, &
 g(z)=&\frac{e^{izx+\tau \int_\R(e^{izy}-1-iz(e^y-1))\nu(dy)}}{iz(1+iz)}.                                                                                                                                                      \end{align*}

Thus $\frac{\partial h(z,\sigma)}{\partial \sigma}$ is bounded. On the other hand $\int_{iv+\R} |g(z)| dz<\infty$ because $|\E(e^{-i(iv+s)J_\tau})|\leq \E(e^{vJ_\tau})<\infty$, where $J_\tau$ is the jump part of $X_\tau$. By Lemma \ref{lemma1} we can differentiate under the integral sign. 

Then, 
\begin{align*}
 \V^L_t=&\frac{\partial C_t(x)}{\partial \sigma}= S_t\frac{e^x}{2\pi}\int_{iv+\R} \frac{e^{izx}e^{\tau \Psi_\sigma(-z)}}{iz(iz+1)}\tau  \sigma  iz(1+iz)dz \notag\\
 =& S_t\tau\sigma e^xf_\tau(x) .
\end{align*}


In order to complete the information provided by vega we can calculate the derivative with respect to jumps intensity.

We assume that  $\int_{iv+\R}|e^{\tau \Psi(-z)}|dz<\infty$. Let $\nu(dy)=\lambda \bar\nu(dy)$ with $\lambda \in [\lambda_1,\lambda_2]$, then  let
\begin{align*}
 h(z,\lambda)=&\frac{e^{\tau \lambda \int_{\R}\big[e^{-izy}-1+iz(e^y-1)\big]\bar\nu(dy)}}{\tau \int_{\R}\big[e^{-izy}-1+iz(e^y-1)\big]\bar\nu(dy)},\\
 g(z)=&\frac{e^{(iz+1)x}e^{\tau iz(1+iz)\frac{\sigma^2}{2}}}{iz(iz+1)}\tau\int_{\R}\big[e^{-izy}-1+iz(e^y-1)\big]\bar\nu(dy),
\end{align*}
where $\frac{\partial h(z,\lambda)}{\partial \lambda}$ is bounded and from Lemma \ref{lemma3} $\int_{iv+\R}|g(z)|dz<\infty$, then  by Lemma \ref{lemma1}

\begin{align*}
 \frac{\partial_\tau C_t(x)}{\partial \lambda}=& \frac{\partial}{\partial \lambda} S_t\frac{1}{2\pi}\int_{iv+\R}\frac{e^{(iz+1)x}e^{\tau\Psi(-z)}}{iz(iz+1)}dz\notag\\=&\tau S_t\frac{1}{2\pi}\int_{iv+\R}\frac{e^{(iz+1)x}e^{\tau\Psi(-z)}}{iz(iz+1)}\overline{\Psi}_J(-z)dz,\notag 
\end{align*}
with $\overline{\Psi}_J(-z)=\int_{\R}\big[e^{-izy}-1+iz(e^y-1)\big]\bar\nu(dy)$. 
Using Fubini's Theorem we obtain

{\allowdisplaybreaks\begin{align*}
 \frac{\partial_\tau C_t(x)}{\partial \lambda}=& \frac{\partial}{\partial \lambda} S_t\frac{1}{2\pi}\int_{iv+\R}\frac{e^{(iz+1)x}e^{\tau\Psi(-z)}}{iz(iz+1)}dz\notag \\
 =&\tau S_t\frac{1}{2\pi}\int_{iv+\R}\frac{e^{(iz+1)x}e^{\tau\Psi(-z)}}{iz(iz+1)}\int_{\R}\big[e^{-izy}-1+iz(e^y-1)\big]\bar\nu(dy)dz\notag \\
 =&\tau S_t\Big[\int_\R \Big(e^{y} \frac{e^{x-y}}{2\pi}\int_{iv+\R}e^{iz(x-y)}\frac{e^{\tau\Psi(-z)}}{iz(1+iz)}dz \notag \\
 &\rfill{-\frac{e^{x}}{2\pi}\int_{iv+\R}e^{izx}\frac{e^{\tau\Psi(-z)}}{iz(1+iz)}dz }    \notag \\
&\rfill{+ (e^y-1)\frac{e^{x}}{2\pi}\int_{iv+\R}e^{izx}\frac{e^{\tau\Psi(-z)}}{1+iz}dz \Big)\bar\nu(dy)\Big]}\notag \\
 =&\tau \Big[\int_\R \Big(e^{y}C_t(x-y)-C_t(x)-S_t(e^{y}-1)\tilde{\Q}(X_\tau>x)\Big)\bar\nu(dy)\Big]. 
\end{align*}}
The use of Fubini's Theorem is justified by \eqref{cotaI} and the additional hypothesis $\int_{iv+\R}|e^{\tau\Psi(-z)}|dz<\infty$.

\subsubsection*{Theta}

We assume that $\int_{iv+\R}|z^2e^{\tau\Psi(-z)}|dz<\infty$, $\tau\in [\mathcal{T}_1,\mathcal{T}_2]$ and $z\in iv+\R$, let

\begin{align*}
 h(z,\tau)={e^{(iz+1)x_\tau}e^{\tau \Psi(-z)}},&& g(z)=\frac{1}{iz(1+iz)}.
\end{align*}
Then, $\int_{iv+\R}|g(z)|dz<\infty$, moreover, from \eqref{cotaPsi} and   $\int_{iv+\R}|z^2e^{\tau\Psi(-z)}|dz<\infty$
\begin{align*}\frac{\partial h(z, \tau) }{\partial \tau}=& {e^{(iz+1)x_\tau}e^{\tau \Psi(-z)}}\Big(-r(1+iz)+\Psi(-z)\Big)
\end{align*}
is bounded and by Lemma \ref{lemma1}, 

\begin{align*}
  \Theta_t^L=&\frac{\partial_\tau C_t(x_\tau)}{\partial \tau}= \frac{\partial}{\partial \tau} S_t\frac{1}{2\pi}\int_{iv+\R}\frac{e^{(iz+1)x_\tau}e^{\tau\Psi(-z)}}{iz(iz+1)}dz\notag \\
 =& S_t\frac{1}{2\pi}\int_{iv+\R}\frac{e^{(iz+1)x_\tau}e^{\tau\Psi(-z)}}{iz(1+iz)}\Big(\Psi(-z)-r(1+iz)\Big)dz.\notag \\
\end{align*}

Using Fubini's Theorem we obtain
{\allowdisplaybreaks\begin{align*}
 \Theta_t^L=&S_t\frac{1}{2\pi}\int_{iv+\R}\frac{e^{(iz+1)x_\tau}e^{\tau\Psi(-z)}}{iz(1+iz)}\Big(\Psi(-z)-r(1+iz)\Big)dz\notag \\
 =& S_t\Big[-\frac{r}{2\pi}\int_{iv+\R}\frac{e^{(iz+1)x_\tau}e^{\tau\Psi(-z)}}{iz}dz \notag \\
 &+  \frac{1}{2\pi}\int_{iv+\R}\frac{e^{(iz+1)x_\tau}e^{\tau\Psi(-z)}}{iz(iz+1)}\Big(iz(1+iz)\frac{\sigma^2}{2}\notag \\
 &\rfill{+\int_{\R}\big[e^{-izy}-1+iz(e^y-1)\big]\nu(dy)\Big)dz\Big]}\notag \\
 =&S_t\Big[re^{x_\tau}\Q(X_\tau>x_\tau)+\frac{\sigma^2}{2}e^{x_\tau}{f}_\tau(x_\tau)\notag \\
& \rfill{ +\int_\R \Big(e^{y} \frac{e^{x-y}}{2\pi}\int_{iv+\R}e^{iz(x-y)}\frac{e^{\tau\Psi(-z)}}{iz(1+iz)}dz -\frac{e^{x}}{2\pi}\int_{iv+\R}e^{izx}\frac{e^{\tau\Psi(-z)}}{iz(1+iz)}dz     }\notag \\
&\rfill{+ (e^y-1)\frac{e^{x}}{2\pi}\int_{iv+\R}e^{izx}\frac{e^{\tau\Psi(-z)}}{1+iz}dz \Big)\nu(dy)\Big]}\notag \\
 =&S_t\Big[re^{x_\tau}\Q(X_\tau>x_\tau)+\frac{\sigma^2}{2}e^{x_\tau}{f}_\tau(x_\tau)\Big]\notag \\
&\rfill{ +\int_\R \Big(e^{y}C_t(x_\tau-y)-C_t(x_\tau)-S_t(e^{y}-1)\tilde{\Q}(X_\tau>x_\tau)\Big)\nu(dy).} 
\end{align*}}
The use of Fubini's Theorem is justified by \eqref{cotaI} and the additional hypothesis $\int_{iv+\R}|z^2e^{\tau\Psi(-z)}|dz<\infty$.

\subsection{Second Order Greeks for Call Options}


\subsubsection*{Gamma}
Once Delta is obtained, we must only diferentiate again with respect to $S_t$, to obtain Gamma. We assume that $X_\tau$ has density $f$ and $S_t\in [A,B]$, then

\begin{align*}
 \Gamma_{t}^L=&\frac{\partial^2 C_{t}(x(S_t))}{\partial S_t^2}= \frac{\partial \tilde{\Q}(X_\tau>x(S_t))}{\partial S_t}\notag\\
=& \frac{1}{S_t}\tilde{f}(x)= \frac{e^{x}}{S_t}{f}(x). 
  \end{align*} 

\subsubsection*{Vanna}
We assume that $\int_{iv+\R} |ze^{\tau\Psi(-z)}|dz<\infty$ and $0<\Sigma_1\leq \sigma\leq \Sigma_2$, then

\begin{align*}
 \frac{\partial^2 C_t(x)}{\partial \sigma \partial S_t}=& \frac{\partial \V_t^L}{\partial S_t}=\tau \sigma e^{x(S_t)}{f}_t(x(S_t))-\tau \sigma e^{x(S_t)}\Big({f}_t(x(S_t))+{f}_t'(x(S_t))\Big)\notag\\
 =&-\tau \sigma e^{x} {f}_\tau'(x).
\end{align*}

\subsubsection*{Vomma}
We assume that $\int_{iv+\R} |z^2e^{\tau\Psi(-z)}|dz<\infty$ and $0<\Sigma_1\leq \sigma\leq \Sigma_2$, let $z\in iv+\R$ and  denote

\begin{equation*}
 h(z,\sigma)=z^2e^{\tau iz(1+iz)\frac{\sigma^2}{2}},\quad
 g(z)= \frac{e^{izx+\tau \int_\R(e^{izy}-1-iz(e^y-1))\nu(dy)}}{z^2}.
\end{equation*}

Thus $\frac{\partial h(z,\sigma)}{\partial \sigma}$ is bounded and $\int_{iv+\R} |g(z)| dz<\infty,$ 
because $|\E(e^{-izJ_\tau})|\leq \E(e^{vJ_\tau})<\infty$, where $J_\tau$ is the jump part of $X_\tau$. By Lemma \ref{lemma1} we can differentiate under the integral sign.

\begin{align*}
 \frac{\partial^2 C_t(x)}{\partial \sigma^2}=& \frac{\partial \V_t^L}{\partial \sigma}=S_t\tau e^{x}{f}_\tau(x) + S_t\tau \sigma \frac{e^x}{2\pi}\int_{iv+\R}e^{izx}e^{\tau \Psi(-z)} \tau\sigma (iz-z^2)dz \notag\\
 =& S_t\tau e^{x}\Big({f}_\tau(x) +\tau \sigma^2 \big[f_\tau'(x)+f_\tau''(x)\big]\Big).
\end{align*}

\subsubsection*{Charm}
We assume that  $\int_{iv+\R} |z^3e^{\tau\Psi(-z)}|dz<\infty$, $\tau\in [\mathcal{T}_1,\mathcal{T}_2]$ and $z\in iv+\R$, let  

\begin{align*}
 h(z,\tau)=z{e^{(iz+1)x_\tau}e^{\tau \Psi(-z)}},\qquad g(z)=\frac{1}{z(1+iz)}.
\end{align*}
Then, $\int_{iv+\R}|g(x)|dx<\infty$  and by Lemma \ref{lemma3}

\begin{align*}\frac{\partial h(z, \tau) }{\partial \tau}=& z{e^{(iz+1)x_\tau}e^{\tau \Psi(-z)}}\Big(-r(1+iz)+\Psi(-z)\Big)
\end{align*}
is bounded. By Lemma \ref{lemma1}, 

\begin{align}
 \frac{\partial^2 C_t(x)}{\partial \tau \partial S_t}=& \frac{\partial \tilde\Q(X_\tau>x_\tau)}{\partial \tau}= \frac{\partial }{\partial \tau} \frac{-1}{2\pi}\int_{iv+\R} e^{(iz+1)x_\tau}\frac{e^{\tau \Psi(-z)}}{1+iz}dz \notag \\
=&\frac{1}{2\pi}\int_{iv+\R} e^{(iz+1)x_\tau}\frac{e^{\tau \Psi(-z)}}{1+iz}\Big(r(1+iz)-\Psi(-z) \Big)dz. \label{charmL2}
\end{align}
Using Fubini's Theorem  we obtain
{\allowdisplaybreaks\begin{align}
 \frac{\partial^2 C_t(x)}{\partial \tau \partial S_t}=&\frac{1}{2\pi}\int_{iv+\R} e^{(iz+1)x_\tau}\frac{e^{\tau \Psi(-z)}}{1+iz}\Big(r(1+iz)-\Psi(-z) \Big)dz \notag\\
=& re^{x_\tau} f_\tau(x_\tau) -\frac{\sigma^2}{2}e^{x_\tau} f_\tau'(x_\tau)\notag\\
&-\int_\R\Big[e^y\frac{e^{x_\tau-y}}{2\pi}\int_{iv+\R} e^{iz(x_\tau-y)}\frac{e^{\tau \Psi(-z)}}{1+iz}dz -\frac{e^{x_\tau}}{2\pi}\int_{iv+\R} e^{izx_\tau}\frac{e^{\tau \Psi(-z)}}{1+iz}dz\notag \\
 & +(e^y-1)\Big\{\frac{e^{x_\tau}}{2\pi}\int_{iv+\R} e^{izx_\tau}{e^{\tau \Psi(-z)}}dz\notag\\
 &\rfill{-\frac{e^{x_\tau}}{2\pi}\int_{iv+\R} e^{izx_\tau}\frac{e^{\tau \Psi(-z)}}{1+iz}dz \Big\}\Big]\nu(dy)} \notag \\ 
 =& -re^{x_\tau} f_\tau(x_\tau) +\frac{\sigma^2}{2}e^{x_\tau} f_\tau'(x_\tau) \notag \\
 &+\int_\R\Big[-e^y\tilde{\Q}(X_\tau>x_\tau-y)+\tilde\Q(X_\tau>x_\tau)\notag \\
 &\rfill{+(e^y-1)\big\{e^xf_\tau(x)+\tilde\Q(X_\tau>x_\tau)\big\}\Big]\nu(dy)} \notag \\
 =& re^{x_\tau} f_\tau(x_\tau) -\frac{\sigma^2}{2}e^{x_\tau} f_\tau'(x_\tau) \notag \\
 & -\int_\R\Big[e^y\Big(\tilde\Q(X_\tau>x_\tau)-\tilde\Q(X_\tau>x_\tau-y)\Big)\notag\\
 &\rfill{+(e^y-1)e^{x_\tau} f_\tau(x_\tau)\Big]\nu(dy).}\notag \\ \label{charmL}
\end{align}}
The use of Fubini's Theorem is justified by \eqref{cotaI} and the additional hypothesis $\int_{iv+\R}|z^3e^{\tau\Psi(-z)}|dz<\infty$.

\subsubsection*{Veta}

We assume that $\int_{iv+\R} |z^4e^{\tau\Psi(-z)}|dz<\infty$. Similar to $\text{Charm}_t$, we assume that $\tau\in [\mathcal{T}_1,\mathcal{T}_2]$ and $z\in iv+\R$, and denote  
\begin{align*}
 h(z,\tau)=z^2{e^{(iz+1)x_\tau}e^{\tau \Psi(-z)}},\quad g(z)=\frac{1}{z^2}.
\end{align*}
Then, $\int_{iv+\R}|g(z)|dz<\infty$ and by Lemma \ref{lemma3}

\begin{align*}\frac{\partial h(z, \tau) }{\partial \tau}=& z^2{e^{(iz+1)x_\tau}e^{\tau \Psi(-z)}}\Big(-r(1+iz)+\Psi(-z)\Big)
\end{align*}
is bounded. By Lemma \ref{lemma1} we can differentiate under the integral sign,  

\begin{align}
 \frac{\partial^2 C_t(x_r)}{\partial \sigma \partial \tau}=& \frac{\partial \V_t^L}{\partial \tau}=\frac{\partial S_t\tau \sigma e^{x_\tau}f_\tau(x_\tau)}{\partial \tau} \notag\\
=& S_t \sigma \Big[e^{x_\tau}f_\tau(x_\tau)-r\tau e^{x_\tau}f_\tau(x_\tau) + \frac{\tau e^{x_\tau}}{2\pi}\int_{iv+\R}\frac{\partial}{\partial \tau}e^{izx_\tau}e^{\tau\Psi(-z)}dz \Big]\notag\\
=&S_t \sigma \Big[e^{x_\tau}f_\tau(x_\tau)-r\tau e^{x_\tau}f_\tau(x_\tau)\notag\\ 
&\rfill{+ \frac{\tau e^{x_\tau}}{2\pi}\int_{iv+\R} e^{izx_\tau}e^{\tau\Psi(-z)}\Big(\Psi(-z)-riz\Big)dz\Big].}\notag \\ \label{VetaL2}
\end{align}
Using Fubini's Theorem, we obtain

{\allowdisplaybreaks\begin{align}
 \frac{\partial^2 C_t(x_r)}{\partial \sigma \partial \tau}=& S_t \sigma \Big[e^{x_\tau}f_\tau(x_\tau)-r\tau e^{x_\tau}f_\tau(x_\tau)\notag\\ 
 &\rfill{+ \frac{\tau e^{x_\tau}}{2\pi}\int_{iv+\R} e^{izx_\tau}e^{\tau\Psi(-z)}\Big(\Psi(-z)-riz\Big)dz} \notag \\
 =& S_t \sigma e^{x_\tau}\Big[f_\tau(x_\tau)-r\tau \big[f_\tau(x_\tau)+f_\tau'(x_\tau)\big] \notag\\
 &+\frac{\tau }{2\pi}\int_{iv+\R}e^{izx_\tau}e^{\tau\Psi(-z)}\big\{\frac{\sigma^2}{2}(iz-z^2)\notag \\
 &\rfill{+\int_\R\big(e^{-izy}-1+iz(e^y-1)\big)\nu(dy)\big\}dz \Big]}\notag\\
=& S_t \sigma e^{x_\tau}\Big[f_\tau(x_\tau)-r\tau \big[f_\tau(x_\tau)+f_\tau'(x_\tau)\big]+ \tau\frac{\sigma^2}{2}\big[f_\tau'(x_\tau)+f_\tau''(\tau)\big]\notag\\
 &\rfill{+\tau \int_\R\Big(f_\tau(x_\tau-y)-f_\tau(x_\tau)+(e^y-1)f_\tau'(x_\tau)\Big)\nu(dy) \Big].}\notag\\ \label{VetaL}
\end{align}}
The use of Fubini's Theorem is justified by \eqref{cotaI} and the additional hypothesis $\int_{iv+\R}|z^4e^{\tau\Psi(-z)}|dz<\infty$.

\subsubsection*{Vera}

Assuming that $\int_{iv+\R}|ze^{\tau\Psi(-z)}|dz<\infty$, and $0<\Sigma_1\leq \sigma\leq \Sigma_2$,
\begin{align*}
 \frac{\partial^2 C_t(x_r)}{\partial \sigma \partial r}=& \frac{\partial \V_t^L}{\partial r}= S_t\tau \sigma e^{x_r}\Big(-\tau f_\tau(x_r)-\tau f_\tau'(x_r)\Big) \notag  \\
 =&-S_t\tau^2 \sigma e^{x_r}\Big( f_\tau(x_r)+ f_\tau'(x_r)\Big). 
\end{align*}

\subsection{Third Order Greeks for Call Options}

\subsubsection*{Color}

We assume that $\int_{iv+\R}|z^4e^{\tau\Psi(-z)}|dz<\infty$, $\tau\in [\mathcal{T}_1,\mathcal{T}_2]$. Let $z\in iv+\R$ and  

\begin{align*}
 h(z,\tau)=&z^2{e^{(iz+1)x_\tau}e^{\tau \Psi(-z)}}\\ g(z)=&\frac{1}{z^2}.
\end{align*}

Then, $\int_{iv+\R}|g(x)|dx<\infty$ and by Lemma \ref{lemma3}
\begin{align*}\frac{\partial h(z, \tau) }{\partial \tau}=& z^2{e^{(iz+1)x_\tau}e^{\tau \Psi(-z)}}\Big(-r(1+iz)+\Psi(-z)\Big)
\end{align*}
is bounded. By Lemma \ref{lemma1} we can differentiate under the integral sign.

Thus,

\begin{align}
 \frac{\partial^3 C_t(x)}{\partial S_t^2 \partial \tau}=& \frac{\partial \Gamma_t^L}{\partial \tau}=\frac{1}{S_t 2\pi}\int_{iv+\R}e^{(iz+1)x_\tau}e^{\tau \Psi(-z)}\Big(-r(iz+1)+\Psi(-z)\Big)dz.  \label{ColorL2}
\end{align}

Using Fubini's Theorem we obtain
\begin{align}
 \frac{\partial^3 C_t(x)}{\partial S_t^2 \partial \tau}=& \frac{1}{S_t 2\pi}\int_{iv+\R}e^{(iz+1)x_\tau}e^{\tau \Psi(-z)}\Big(-r(iz+1)+\Psi(-z)\Big)dz \notag \\
 =&\frac{e^x}{S_t}\Big[ -r\big(f(x)+f'(x)\big)+\frac{\sigma^2}{2}\big(f'(x)+f''(x)\big) \notag \\
 &\rfill{ + \frac{1}{2\pi}\int_{iv+\R} e^{izx}e^{\tau \Psi(-z)}\int_\R e^{-izy}-1+iz(e^y-1)\nu(dy)dz\Big] }\notag\\
 =& -\frac{e^x}{S_t}\Big[r\Big(f_\tau(x)+f'_\tau(x)\Big)-\frac{\sigma^2}{2}\Big(f'_\tau(x)+f''_\tau(x)\Big) \notag \\
 &\rfill{ +\int_\R \Big(f_\tau(x)-f_\tau(x-y)-(e^y-1)f'_\tau(x)\Big)\nu(dy)\Big]. }\label{ColorL}
\end{align}
Fubini is justified by \eqref{cotaI} and the hypothesis $\int_{iv+\R}|z^4e^{\tau\Psi(-z)}|dz<\infty$.

%
%

\subsubsection*{Speed}
Assuming that $\int_{iv+\R}|ze^{\tau\Psi(-z)}|dz<\infty$,
\begin{align*}
 \frac{\partial^3 C_t(x_r)}{\partial S_t^3 }=& \frac{\partial \Gamma_t^L}{\partial S_t}=\frac{e^{x(S_t)}\Big(-\frac{1}{S_t} f_\tau(x(S_t))-\frac{1}{S_t}f_\tau'(x(S_t)) \Big)S_t-e^{x(S_t)}f_\tau(x(S_t))}{S_t^2}\notag \\
 =&-\frac{e^{x}}{S^2_t}\Big(2f_\tau(x)+f_\tau'(x)\Big). 
\end{align*}

\subsubsection*{Ultima}
We assume that $\int_{iv+\R}|z^6e^{\tau\Psi(-z)}|dz<\infty$. First we calculate $\frac{\partial f^{(n)}_\tau (x)}{\partial \sigma}$  for $n=0,1,2$. For $0<\Sigma_1\leq \sigma\leq \Sigma_2$ and $z\in iv+\R$, and denote 
\begin{align*}
h_n(z,\sigma)=(iz)^{n+2}{e^{\tau \Psi(-z)}},\quad g(z)=-\frac{e^{izx}}{z^2}.                                                                                                                                                      \end{align*}
Thus, $\int_{iv+\R} |g(z)| dz<\infty$ and  $\frac{\partial h_n(z,\sigma)}{\partial \sigma}$ is bounded for $n=0,1,2$. By Lemma \ref{lemma1} we can differentiate under the integral sign. Then,

\begin{align}
 \frac{\partial^n f_\tau (x)}{\partial \sigma}=&\frac{\partial }{\partial \sigma}\frac{1}{2\pi}\int_{iv+\R} (iz)^ne^{izx}e^{\tau\Psi(-z)}dz \notag \\
 =&\tau \sigma\frac{1}{2\pi}\int_{iv+\R} [(iz)^{n+1}-(iz)^{n+2}]e^{izx}e^{\tau\Psi(-z)}dz \notag \\
=& \tau \sigma\Big(f_\tau^{(n+1)}(x)+f_\tau^{(n+2)}(x)\Big) \label{dfdsigma}.
\end{align}

%
%
%
%

Now, we have

\begin{align*}
 \frac{\partial^3C_t(x)}{\partial \sigma^3 }=& \frac{\partial S_t \tau e^x\Big( f_\tau(x)+\tau\sigma^2\big[f_\tau'(x)+f_\tau''(x)\big]\Big)}{\partial \sigma}\notag \\
 =& S_t\tau^2 \sigma e^x \Big(3\big(f_\tau'(x)+f_\tau''(x)\big)+\tau\sigma^2\big[f_\tau''(x)+2f_\tau'''(x)+f_\tau^{(iv)}(x)\big]    \Big) .
\end{align*}

\subsubsection*{Zomma}
We assume that $\int_{iv+\R}|z^2e^{\tau\Psi(-z)}|dz<\infty$ and $0<\Sigma_1\leq \sigma\leq \Sigma_2$. Then, 

\begin{align*}
 \frac{\partial^3 C_t(x_r)}{\partial S_t^2 \partial \sigma }=& \frac{\partial Vanna_t^L}{\partial S_t}=\frac{\partial \tau \sigma e^{x(S_t)} {f}_\tau'(x(S_t))}{\partial S_t}\notag \\
 =& -\frac{\tau \sigma e^{x}}{S_t} \Big({f}_\tau'(x)+{f}_\tau''(x) \Big).
\end{align*}

\section{Examples}

\subsection{The Black-Scholes Model} 

If we assume that the gaussian distribution and density are exactly computed in R software, 
we can compare the Greeks for  Black-Scholes model using Lewis representation.

To aproximate the Fourier Transform we cut the integral  between $-A/2$ and $A/2$ and take a uniform partition of $[-A/2,A/2]$ of size $N$:

$$\int_{\R} e^{izx}g(z)dz\approx \int_{-A/2}^{A/2} e^{izx}g(z)dz\approx \frac{A}{N}\sum_{k=0}^{N-1}w_{k} e^{iz_k x}g(z_k),$$
where $z_k=-\frac{A}{2}+k\frac{A}{N-1}$ and $w_k$ are weights that correspond to the integration numerical rule.

Table \ref{comp} shows the $\ell_\infty$-errors in Black-Scholes model via Lewis representation and Fast Fourier Transform using: $S_t=1$, $r=0.05$, $T=1$, $\sigma=0.1$,  $A=300$ and $N=2^{22}$. The $\ell_\infty$-errors are $$\ell_\infty\mbox{-error}(GL)=\max_{x\in [-0.7,0.7]} |GL-G|,$$
for $x=\ln(K/S_t)-r\tau$.

\begin{table}[ht]
\centering
\begin{tabular}{rcl}
  \hline
  Greek & Expression & $\ell_\infty$-error\\
   \hline
  Call & $C=S\E(e^{X_\tau}-e^x)^+$ & 1.2e-07 \\ 
  Delta & $ \partial_S {C}(x)$ & 2.4e-07 \\ 
  Rho & $\partial_r C_{t}(x)$&1.9e-07 \\ 
  Vega & $ \partial_\sigma {C}(x)$&9.5e-08 \\ 
  Theta & $\partial_\tau {C}(x)$&1.2e-08 \\ 
  Gamma & $\partial^2_{SS} {C}(x)$& 9.5e-07 \\ 
  Vanna & $\partial^2_{\sigma S} {C}(x)$ &6.3e-07 \\ 
  Vomma & $\partial^2_{\sigma \sigma} {C}(x)$&7.5e-07 \\ 
  Charm & $\partial^2_{S\tau} {C}(x)$&6.8e-08 \\ 
  Veta & $\partial^2_{\sigma\tau} {C}(x)$&8.9e-08 \\ 
  Vera &$\partial^2_{\sigma r} {C}(x)$& 5.8e-07 \\ 
  Color & $\partial^3_{SS\tau} {C}(x)$&5.6e-07 \\ 
  Speed & $\partial^3_{SSS} {C}(x)$&6.3e-06 \\ 
  Ultima & $\partial^3_{\sigma\sigma\sigma} {C}(x)$&1.2e-05 \\ 
  Zomma & $\partial^3_{SS\sigma} {C}(x)$&9.5e-06 \\ 
   \hline
\end{tabular}
\caption{$\ell_\infty$-errors in Black-Scholes model via Lewis representation and Fast Fourier Transform using: $S_t=1$, $r=0.05$, $T=1$, $\sigma=0.1$, $A=300$ and $N=2^{22}$.}
\label{comp}
\end{table}

\subsection{The Merton model}

In this section we show some results for the Merton model. The Merton model has four parameters $(\sigma,\mu_J,\sigma_J,\lambda)$ where $\sigma$ is the  diffusion parameter, $\lambda$ is the jump intensity, $\mu_J$ and $\sigma_J$ are the mean and standard deviation of the jump which are gaussianly distributed. The characteristic function for the Merton model is:

\begin{align}
 \E(e^{izX_T})=& \exp\left\{iz\Big[\frac{\sigma^2}{2}-\lambda\big(e^{\mu_J+\frac{\sigma_J^2}{2}}-1\big)\Big]+z^2\frac{\sigma^2}{2}+\lambda\big(e^{iz\mu_J-z^2\frac{\sigma_J^2}{2}}-1\big)\right\}. \label{FCMerton}
\end{align}

All Greeks for \emph{At The Money} ($K=S_0e^{-rT}$) are shown in Table \ref{GMerton} following section \ref{GCALL}. Here we took $A=500$, $N=2^{20}$ and $A=500$, $N=2^{22}$, the  $\ell_\infty$-error for $x\in [-0.7,0.7]$ is in all Greeks lower than $10^{-5}$. In Figure \ref{GragMerton} the curves are shown in terms of $x=\ln(K/S_0)-rT$ for all Greeks with the comparison of the Black-Scholes model with volatility equal to implied volatility \emph{At The Money}.

\begin{table}[ht]
\centering
\begin{tabular}{rcccc}
  \hline
 && $A=500$, $N=2^{20}$ & $A=500$, $N=2^{21}$ & error \\ 
  \hline
  Call & $C$ & 0.0547129 & 0.0547129 &  2.6e-08 \\ 
  Delta &$ {\partial_S C}$& 0.5273560 & 0.5273562 &  2.5e-07\\ 
  Rho & ${\partial_r C}$&0.4726431 & 0.4726433 &  2.2e-07\\ 
  Vega &$ {\partial_\sigma C}$ &0.3077754 & 0.3077755 &  1.5e-07 \\ 
  Theta & ${\partial_\tau C}$&0.0524286 & 0.0524286 & 2.5e-08\\ 
  Gamma & ${\partial^2_{SS} C}$&3.0777536 & 3.0777550 &  1.5e-06 \\ 
  Vanna &  ${\partial^2_{\sigma S}C}$ &0.1538877 & 0.1538878 & 7.3e-08 \\ 
  Vomma & ${\partial^2_{\sigma\sigma} C}$&0.9091776 & 0.9091780 & 4.3e-07 \\ 
  Charm &${\partial^2_{S\tau} C}$& 0.1682859 & 0.1682860 & 8.1e-08\\ 
  Veta & ${\partial^2_{\sigma\tau}  C}$&0.1222075 & 0.1222076 &  5.8e-08 \\ 
  Vera & ${\partial^2_{\sigma r} C}$&-0.1538877 & -0.1538878 & 7.3e-08\\ 
  Color & ${\partial^3_{SS\tau} C}$&1.8556786 & 1.8556795 &  8.8e-07 \\ 
  Speed &${\partial^3_{SSS} C}$& -4.6166303 & -4.6166325 & 2.2e-06 \\ 
  Ultima &  ${\partial^3_{\sigma\sigma\sigma} C}$&-11.5390901 & -11.5390956 & 5.5e-06 \\ 
  Zomma & ${\partial^3_{SS\sigma} C}$&-21.6857596 & -21.6857699 & 1.0e-05\\ 
   \hline
\end{tabular}
\caption{Greeks in Merton model with: $S_0=1$, $r=0.05$, $x=0$, $T=1$, $\sigma=0.1$, $\mu_J=-0.005$, $\sigma_J=0.1$, $\lambda=1$.}
\label{GMerton}
\end{table}

\begin{figure}
 \centering
 \includegraphics[scale=.67,keepaspectratio=true]{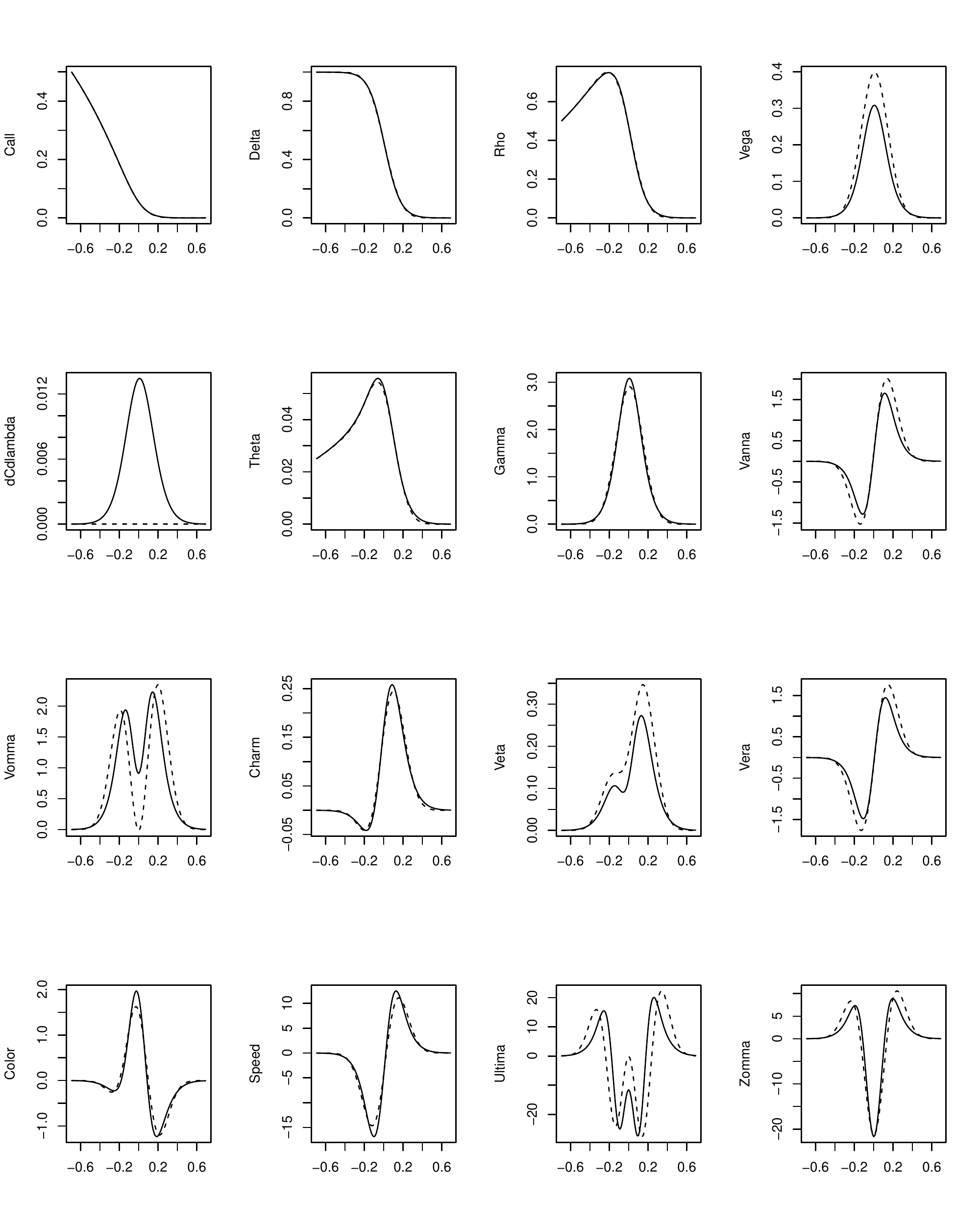}
 \caption{Greeks  in terms of $x=\ln(K/S_0)-rT$ for the Merton Model with parameters equal to Table \ref{GMerton} (continuos line). Discontinuos line: Black-Scholes Model with volatility equal to implied volatility in $x=0$ ($\sigma_{imp}(0)\approx 0.137$).}
 \label{GragMerton}
\end{figure}

 The characteristic function in this case is \eqref{FCMerton}. 
 To compute sensitivities w.r.t. $\mu_J$, $\sigma_J$ and $\lambda$ 
 we only need to differentiate the characteristic exponent with respect to these parameters:
{\allowdisplaybreaks\begin{align*}
\frac{\partial \Psi(-z)}{\partial \mu_j}=&\lambda iz\Big[e^{\mu_J+\sigma_J^2/2}-e^{-iz\mu_J-z^2\sigma_J^2/2}\Big],\\ 
\frac{\partial \Psi(-z)}{\partial \sigma_j}=&\lambda \sigma_J\Big[ize^{\mu_J+\sigma_J^2/2}-z^2e^{-iz\mu_J-z^2\sigma_J^2/2}\Big],\\
\frac{\partial \Psi(-z)}{\partial \lambda_j}=&iz\Big[e^{\mu_J+\sigma_J^2/2}-1\Big]+e^{-iz\mu_J-z^2\sigma_J^2/2}-1,
\end{align*}}
and for $\theta=\mu_J, \sigma_J, \lambda,$ 
\begin{align*}
     \frac{\partial C_\theta(x)}{\partial \theta }=&\tau S_t\frac{e^x}{2\pi}\int_{iv+\R}e^{izx}\frac{e^{\tau \Psi_\theta(-z)} }{iz(1+iz)}\frac{\partial \Psi_\theta(-z)}{\partial \theta}dz.
    \end{align*}

The differentiation under the integral sign is justified as above.    
    
Using the same parameters presented in Table \ref{GMerton} we obtain the sensitivities for ATM given in Table \ref{SMerton}.

\begin{table}[ht]
\centering
\begin{tabular}{rccc}
  \hline
 & $A=500$, $N=2^{20}$ & $A=1000$, $N=2^{22}$& error \\ 
  \hline
  $\mu_J$-sensitivity& 0.006703850 & 0.006703855 & 4.7e-09  \\ 
 $\sigma_J$-sensitivity& 0.239001059 & 0.239001230 & 1.7e-07  \\ 
  $\lambda$-sensitivity& 0.013407701 & 0.013407711 & 9.6e-09 \\  
   \hline
\end{tabular}
\caption{Sensitivities for Merton model with: $S_0=1$, $r=0.05$, $x=0$, $T=1$, $\sigma=0.1$, $\mu_J=-0.005$, $\sigma_J=0.1$, $\lambda=1$.}
\label{SMerton}
\end{table}

In Figure \ref{GrasMerton} we show the Greeks in terms of $x=\ln(K/S_0)-rT$.

\begin{figure}
 \centering
 \includegraphics[scale=.45,keepaspectratio=true]{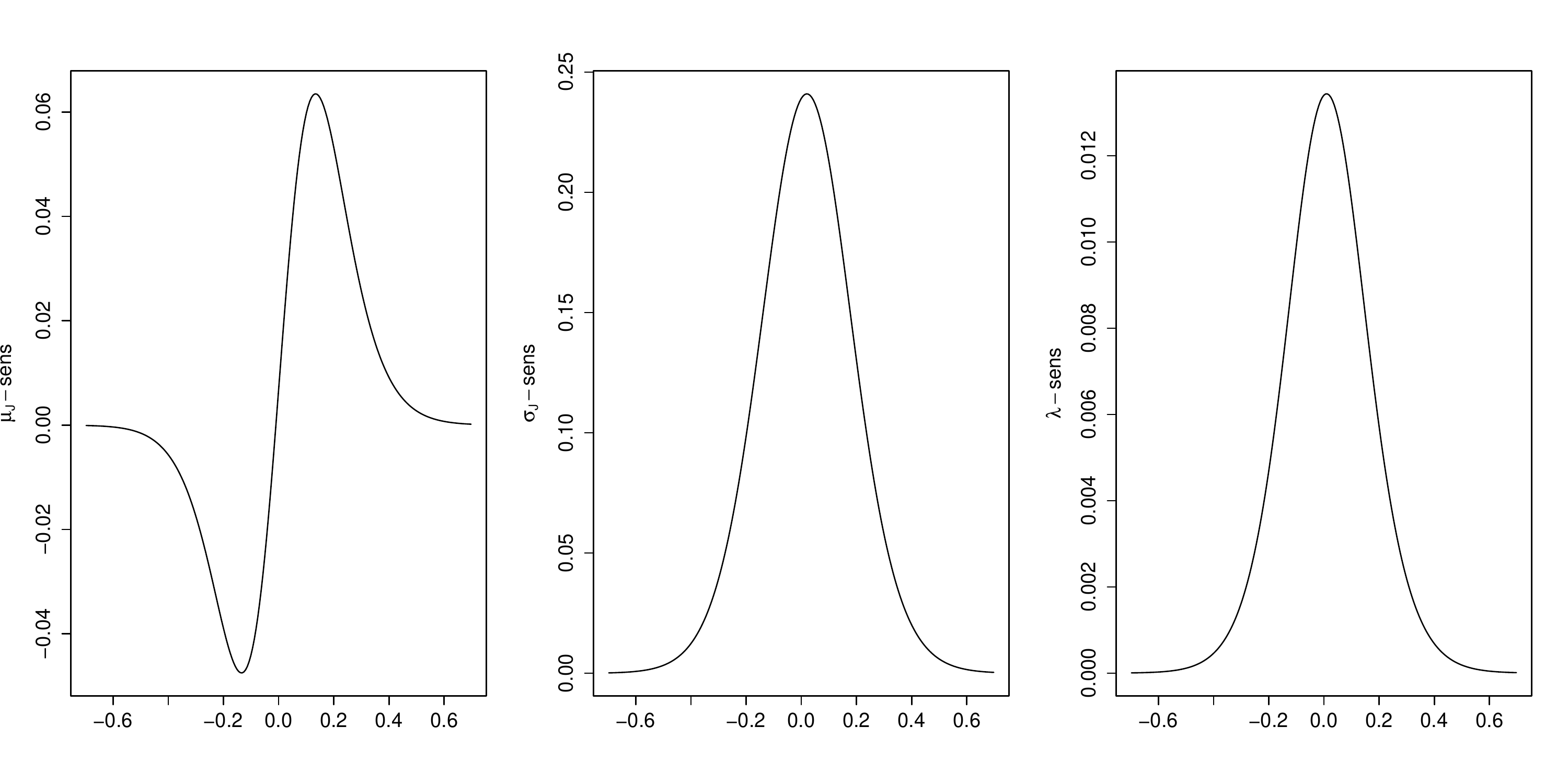}
 \caption{Sensitivities in terms of $x=\ln(K/S_0)-rT$ for Merton Model with parameters equal to Table \ref{SMerton}.}
 \label{GrasMerton}
\end{figure}

In \cite{kienitz} are shown some results for a Digital Option in the Merton model, which were obtained by applying finite difference approximations to the formula for the option prices in \cite{MadanCarrChang98}. Now we will deduce Delta, Gamma and Vega for a Digital Option and thus we will compare the results.

A Digital Option has a payoff given by:

\begin{align*}
 \indicator_{\{S_\tau-K>0\}}=\indicator_{\{X_\tau-x>0\}}.
\end{align*}

Using Lewis representation, the value for a Digital Option is:

\begin{align}
 D(x)=&\Q(X_\tau> x)=-\frac{1}{2\pi}\int_{iv+\R}e^{izx}\frac{e^{\tau\Psi(-z)}}{iz}dz,\label{DigitalC} 
\end{align}

where $x=\ln(K/S_t)-r\tau$. A direct differentiation leads to:

\begin{align}
 \frac{\partial D(x)}{\partial S_t}=& \frac{1}{S_\tau} f_\tau(x),\label{DDelta}\\
 \frac{\partial^2 D(x)}{\partial S_t^2}=& -\frac{1}{S_\tau^2} \Big(f_\tau(x)+f_\tau'(x)\Big),\label{DGamma}\\
\frac{\partial D(x)}{\partial \sigma}=& -\tau \sigma \Big(f_\tau(x)+f_\tau'(x)\Big)\label{DVega}.
\end{align}

Observe that the formulas \eqref{DigitalC}-\eqref{DVega} are valid in general for Digital Call options with $\int_{iv+\R}|ze^{\tau\Psi(-z)}|dz<\infty$. In \eqref{DVega}, differentiation under integral sign is similar to \eqref{dfdsigma}.

Then, our results via FFT are shown in Table \ref{DMerton}. In \cite{kienitz} this values are (by finite difference) : D$=0.531270$, D-Delta$=0.016610$, D-Gamma$=-2.800324\times 10^{-4}$, D-Vega$=-0.560070$. To obtain a given strike we define $\delta=2\pi\frac{N-1}{NA}$.

\begin{table}[ht]
\centering
\begin{tabular}{rcccc}
  \hline
  & D-Call & D-Delta & D-Gamma & D-Vega \\ 
  \hline
 $N=2^{20}$ & 0.531269863 & 0.016610445 & -0.000280032 & -0.560064360 \\ 
  $N=2^{22}$ & 0.531270245 & 0.016610457 & -0.000280032 & -0.560064763 \\ 
  error  & 3.8e-07& 1.2e-08& 2.0e-10& 4.0e-07  \\ 
   \hline
\end{tabular}
\caption{Digital Option and Greeks in Merton model with:$\delta=0.01$, $S_0=100$, $K=100$, $T=1$, $r=0.07$, $\sigma=0.2$, $\mu_J=0.05$, $\sigma_J=0.15$ and $\lambda=0.5$.}
\label{DMerton}
\end{table}

\subsection{The Variance Gamma Model}

In this section we will compare some results from the literature. As an example, in \cite{Glasserman2007} some results are shown for  the Variance Gamma model  with parameters $(\rho,\nu,\theta)$ where the characteristic function is:

\begin{align*}\E[e^{izX_T}]=\exp\left\{\frac{T}{\nu}\Big[iz\ln\big(1-\theta\nu-\frac{\rho^2\nu}{2}\big)-\ln\big(1-iz\theta\nu+\frac{z^2\rho^2\nu}{2}\big)\Big]\right\}.\end{align*}

To obtain a given strike we define $\delta=2\pi\frac{N-1}{NA}$. Thus, in Table \ref{VGcomp} we present two results for $N=2^{20}$ and $N=2^{22}$ with $\delta=0.01$. The error shows the convergence of the complex integral. In \cite{Glasserman2007} these results are obteined by applying finite difference approximations to the formula for the option prices in \cite{MadanCarrChang98}: Call$=11.2669$, Delta$=0.7282$ and $\rho$-derivative$=23.0434$ and in general with LRM method, the error is worse than $10^{-2}.$ 
\begin{table}[ht]
\centering
\begin{tabular}{c|cccc}
  \hline
 &Call & Delta & Gamma & $\frac{\partial Call}{\partial \rho}$\\
    \hline
$N=2^{20}$, $\delta=0.01$ &11.26689113&  0.72818427 & 0.01427437&23.04334371\\  
$N=2^{22}$, $\delta=0.01$& 11.26689919&  0.72818479 & 0.01427438&23.04336021\\ \hline
err$<$ & 8.1e-06& 5.2e-07& 1.0e-08& 1.6e-05 \\ \hline
 \end{tabular}
\caption{Greeks and $\rho$-sensitivity for Variance Gamma model with: $(\rho,\nu,\theta)=(0.2,1,-0.15)$, $r=0.05$, $T=1$, $S_0=K=100$ ($x=-0.05$).}
\label{VGcomp}
\end{table}

%
%


\section{Conclusions}
 
Greeks are an important input for market makers in risk management. A lot of options are path dependent and they don't have explicit formula. However, for the European options in the exponential L\'evy models we have the Lewis formula,  which allows us to obtain closed formulas for Greeks, many of which are only density dependent; others require integration. In general, all Greeks can be approximated   with high  accuracy because they are a simple integral, similar to the Black-Scholes model.

A large numbers of papers are dedicated to obtain Greeks for more complex payoff functions . However, in order to estimate the accuracy of their methods, Greeks approximations are computed through the finite difference technique.

For a fix Strike $K$, we consider $x=\ln(K/S_t)-r\tau$, with $\tau=T-\tau$ the time to maturity. Thus, the Greeks for call options can be calculated through  Table \ref{Greeks}.

\begin{table}
\centering
{\small
\begin{tabular}{||c||r l||} \hline \hline
 First order& & \\ \hline 
Delta & $\partial C_S(x)=$&$\displaystyle\tilde{\Q}(X_\tau>x) $\\
Rho& $\displaystyle\partial C_{r}(x)=$&$\displaystyle\tau Se^{x}\Q(X_\tau>x)$ \\
Vega & $ \displaystyle\partial C_{\sigma}(x)=$&$\displaystyle S\tau\sigma e^xf_\tau(x)$\\
if $\nu=\lambda \bar \nu$ &  $ \displaystyle\partial C_{\lambda}(x)=$& $\displaystyle \tau \Big[\int_\R \Big(e^{y}C(x-y)-C(x)$\\
& & $\displaystyle \qquad -S(e^{y}-1)\tilde{\Q}(X_\tau>x)\Big)\bar\nu(dy)\Big]$\\
Theta & $ \displaystyle\partial C_{\tau}(x)=$&$\displaystyle S\Big[re^{x}\Q(X_\tau>x)+\frac{\sigma^2}{2}e^{x}{f}_\tau(x)\Big]+\frac{\lambda}{\tau}\partial C_{\lambda}(x)$\\\hline 
Second order & & \\ \hline
Gamma&  $\displaystyle \partial^2 C_{SS}(x)=$&$\displaystyle{S^{-1}}{e^{x}}{f_\tau}(x)$\\
Vanna& $\displaystyle \partial^2 C_{\sigma S}(x)=$&$\displaystyle -\tau \sigma e^{x} {f}_\tau'(x)$\\
Vomma& $\displaystyle \partial^2 C_{\sigma\sigma}(x)=$&$\displaystyle S\tau e^{x}\Big({f}_\tau(x) +\tau \sigma^2 \big[f_\tau'(x)+f_\tau''(x)\big]\Big)$ \\
Charm & $ \displaystyle\partial^2 C_{S\tau}(x)=$& see \eqref{charmL2} and \eqref{charmL}\\
Veta &  $ \displaystyle\partial^2 C_{\sigma \tau}(x)=$& see \eqref{VetaL2} and \eqref{VetaL}\\
Vera& $\displaystyle \partial^2 C_{\sigma r}(x)=$&$\displaystyle -S\tau^2 \sigma e^{x}\Big( f_\tau(x)+ f_\tau'(x)\Big)$\\ \hline
Third order& &\\ \hline
Color& $\displaystyle \partial^3 C_{SS\tau}(x)=$& see \eqref{ColorL2} and \eqref{ColorL} \\
Speed& $ \displaystyle \partial^3 C_{SSS}(x)=$& $\displaystyle -{S^{-2}}{e^{x}}\Big(2f_\tau(x)+f_\tau'(x)\Big)$\\
Ultima&  $ \displaystyle\partial^3 C_{\sigma\sigma\sigma}(x)=$& $\displaystyle S\tau^2 \sigma e^x \Big(3\big(f_\tau'(x)+f_\tau''(x)\big)$\\
& &$\displaystyle \qquad \qquad  +\tau\sigma^2\big[f_\tau''(x)+2f_\tau'''(x)+f_\tau^{iv}(x)\big]\Big)$\\
Zomma& $\displaystyle \partial^3 C_{SS\sigma}(x)=$& $\displaystyle -\tau \sigma S^{-1} e^{x} \Big({f}_\tau'(x)+{f}_\tau''(x)\Big)$ \\ \hline \hline
\end{tabular}}
\caption{Greeks in exponential L\'evy models in terms of $x=\ln(K/S)-r\tau$.}
\label{Greeks}
\end{table}


We observe that, if the density of $X_\tau$ is known, then many of the Greeks can be exactly obtained. Some examples of these are: Normal Inverse Gaussian, Variance Gamma, Generalized Hyperbolic, Meixner and others.


\bibliographystyle{econometrica}
\bibliography{catalogo}

\end{document}